\documentclass[twocolumn]{autarc_LH}   

\usepackage{ae}
\usepackage{amssymb,psfrag,balance}
\usepackage[dvipsnames]{xcolor}
\usepackage{amsmath}
\usepackage{amsfonts}
\usepackage{xcolor}
\usepackage{graphicx}

\usepackage{units}

\usepackage{tikz}
\usetikzlibrary{calc,patterns,decorations.pathmorphing,decorations.markings}

\usepackage[export]{adjustbox}

\newcommand{\Rf}{{\mathbb R}}

\newcommand{\Kn}{{\mathbb K}}

\newtheorem{theorem}{Theorem}
\newtheorem{corollary}{Corollary}

\newtheorem{lemma}{Lemma}

\newenvironment{proof}{\begin{trivlist} \item[{ \bf Proof:}] }
{~\hfill$\Box$ \end{trivlist} }

\usepackage[colorlinks=true,allcolors=MidnightBlue]{hyperref}

\begin{document}

\begin{frontmatter}

\title{Multivariable Generalized Super-Twisting Algorithm Robust Control of Linear Time-Invariant Systems}

\thanks[footnoteinfo]{This paper was not presented at any IFAC 
meeting. Corresponding author Eduardo V. L. Nunes.}

\author[UNICAMP]{Jos\'{e} C. Geromel}\ead{geromel@dsce.fee.unicamp.br}, 
\author[UFRJ]{Eduardo V. L. Nunes}\ead{eduardo.nunes@coppe.ufrj.br},       
\author[UFRJ]{Liu Hsu}\ead{lhsu@coppe.ufrj.br}  

\address[UNICAMP]{School of Electrical and Computer Engineering, UNICAMP, SP-Brazil}                                             
\address[UFRJ]{Department of Electrical Engineering, COPPE, Federal University of Rio de Janeiro, RJ-Brazil}   
 

\begin{abstract}              

This paper presents a novel procedure for robust control design of linear time-invariant systems using a Multivariable Generalized Super-Twisting Algorithm (MGSTA). The proposed approach addresses robust stability and performance conditions, considering convex bounded parameter uncertainty in all matrices of the plant state-space realization and Lipschitz exogenous disturbances. The primary characteristic of the closed-loop system, sliding mode finite-time convergence, is thoroughly examined and evaluated. The design conditions, obtained through the proposal of a novel max-type non-differentiable piecewise-continuous Lyapunov function are formulated as Linear Matrix Inequalities (LMIs), which can be efficiently solved using existing computational tools. A fault-tolerant MGSTA control is designed for a mechanical system with three degrees of freedom, illustrating the efficacy of the proposed LMI approach.

\end{abstract}

\end{frontmatter}

\section{Introduction}

The super-twisting algorithm (STA)  \cite{Levant:1993}, \cite{Shtessel_book:2014} is a powerful second-order sliding mode control (SMC) algorithm that reduces chattering by hiding the discontinuous control of conventional SMC \cite{U:92} behind an integrator. This ingenious approach allows for continuous control input, exact compensation of Lipschitz disturbances, and finite-time convergence. The STA has been a prominent subject of research in the SMC community, spurring numerous theoretical and applied studies. Its outreach has been remarkable, being applicable not only for control but also for achieving robust exact differentiators \cite{L:2003}, \cite{Levant2007}  and finite-time observers \cite{Davila2005}, \cite{Floquet2007}, even with unknown inputs.

In \cite{Moreno:2009}, a generalized super-twisting algorithm (GSTA) was introduced to also handle state-dependent perturbations, providing additional theoretical advances. In this paper, we consider uncertain multivariable linear systems. The extension of the scalar STA to the multivariable case (MSTA) follows the original idea introduced in \cite{NE:2014} based on the unit vector switching function.

Several authors have studied various multivariable STA formulations. In \cite{VNH_tac:2017}, a multivariable extension of the scalar variable gain generalized STA from \cite{GMF:2012}, was developed, exploiting key properties of the Jacobian matrix. Further extensions related to the Jacobian matrix are found in \cite{Caamal_Moreno:2019} and  \cite{MROF:2022}. A challenging issue arises when the input matrix, say $B$, is not known, as assumed in \cite{Caamal_Moreno:2019}, \cite{NE:2014} and \cite{VNH_tac:2017}. To address the case of unknown $B$, sufficient conditions were derived in \cite{VNH:2016} to ensure global finite-time stability when $B$ is symmetric and positive definite. This was extended in \cite{KNH:2019}, showing stability for small asymmetry. The MSTA’s sensitivity to changes in the input matrix and the challenge of selecting gains were highlighted. A generalized MSTA was proposed in \cite{MROF:2022} for systems with time-varying and state-dependent input matrices and disturbances, requiring a positive definite symmetric part of $B$.

Only scalar gains were considered in \cite{KNH:2019}, \cite{MROF:2022}, \cite{VNH:2016}. Selecting full matrix gains for the MSTA is challenging. This problem was addressed in \cite{MGM:2024} for systems with time and state-varying uncertain input matrices, establishing  sufficient conditions for stability with full-matrix control gains using a smooth Lyapunov function. However, no systematic method for computing these gains was provided.

In \cite{gero:2023}, MSTA design techniques were generalized for full-matrix gains synthesis. A new LMI-based design was proposed to handle convex bounded model uncertainties in the input matrix and exogenous disturbances with norm-bounded time derivative. Gains can then be designed based on cost weights related to transient and control effort, with an additional parameter for the disturbance bound. 

The results of \cite{gero:2023} were extended in 
\cite{NGH:2024} by incorporating a linear state-dependent term in the MSTA and addressing convex bounded parameter uncertainty in both the state and the input matrices. Additionally, the multivariable extension was not based on the Jacobian matrix but rather on an appropriate scalar function, which can be particularly important for systems with uncertain input matrices. This novel formulation appears to simplify the theoretical analysis. However, this approach requires the number of states to match the number of inputs, significantly limiting the class of linear systems that can be addressed. 

In this paper, we build upon the findings of \cite{NGH:2024} by proposing a novel LMI-based generalized MSTA design for controlling more general linear MIMO plants with internal dynamics, under the assumption of convex bounded parameter uncertainty in all system matrices. The inclusion of internal dynamics significantly increases the complexity of the problem. To address this challenge, we decompose the system into two subsystems and introduce a novel non-differentiable piecewise-continuous Lyapunov function that innovatively combines the Lyapunov functions of each subsystem whose stability conditions emerge from the determination of the one-sided directional derivative \cite{ShePa:1994}. The control design achieves robust performance through the optimization of a guaranteed performance index, ensuring robust global stability. All problems are numerically handled by LMI solvers available in the literature. We develop a fault-tolerant MGSTA controller for a three-degree-of-freedom mechanical system, demonstrating the practical applicability and robustness of the proposed LMI-based approach.

The paper is organized as follows. In the next section, the problem to be dealt with, including convex bounded parameter uncertainty and time derivative norm bounded exogenous uncertainty is presented. In Section $3$, the system structure, composed of the feedback connection of a linear and a nonlinear subsystems is highlighted. The Lyapunov function candidates for each subsystem as well as the composite one that works for the overall system are given. Section $4$, is entirely devoted to robust performance analysis with respect to  exogenous Lipschitz disturbance. In Section $5$, a one-to-one change of variables is used to express the control design conditions as LMIs which makes it possible to solve them by well-known numerical tools. Finite-time convergence is analyzed in Section $6$. In Section $7$, the theory is successfully illustrated by designing a fault-tolerant control of a chain of trailers. Numerical simulations are presented and discussed. Finally, the main conclusions, recommendations and possible follow-up are summarized in Section $8$. 

\subsection{Notation}

For real vectors or matrices, $(^T)$ refers to their transpose. For symmetric matrices, $(\bullet)$ denotes the symmetric block. The symbols $\mathbb{R}$,  $\mathbb{R}_+$, $\mathbb{N}$ and $\Kn$ denote the sets of real, real non-negative, natural numbers and $\Kn=\{1, 2, \cdots, N\}$ with $N \in \mathbb{N} \setminus \{0\}$, respectively. For any symmetric matrix, $X>0$ ($X \geq 0$) denotes a positive (semi-) definite matrix. As usual, for real matrices the symbol $\| \cdot \|$ denotes the maximum singular value and for real vectors, the same symbol denotes the Euclidean norm. The unit simplex $\Lambda \subset \Rf^N$ stands for the set of all nonnegative vectors with the sum of components equal to one. The convex hull of matrices $\{B_i\}_{i \in \Kn}$, with the same dimensions, is denoted by $\text{co} \{B_i\}_{i \in \Kn}$. Finally, $I_n$ denotes the $n \! \times \! n$ identity matrix. Filippov’s definition for the solution of discontinuous differential equations is adopted \cite{F:88}. Time derivative, whenever exists, is denoted as $\dot f(t) = df(t)/dt$. To ease the notation we drop the time argument whenever clear. Finally, $I_n$ denotes the $n \times n$ identity matrix. 

\section{Problem statement}

Consider an uncertain MIMO system in regular form, described by 
\begin{equation}
	\label{eq:main_system}
	\begin{aligned}
		\dot \zeta & = A \zeta + E \sigma \\
		\dot \sigma & = C \zeta  +  D \sigma +  B u + f(t)
	\end{aligned}
\end{equation}
where $\zeta \in \mathbb{R}^{r}$, $\sigma \in \mathbb{R}^n$ are the state variables, $u \in \mathbb{R}^m$ is the control input, $B \in \Rf^{n \times m}$  is the input matrix, and $f(t): \Rf_+ \to \Rf^n$ is the uncertain disturbance. With respect to it, the following assumptions are made:
\begin{description}
    \item[(A1)] The constant matrices $(A, E, C, D, B)$ of compatible
    dimensions are uncertain. They are elements of the compact convex polyhedral set $\mathbb{M}$ generated by the convex combination of the vertices, that is $\mathbb{M} =  \text{co} \{ (A, E, C, D, B)_i \}_{i \in \mathbb{K}}$. 

    \item[(A2)] Each uncertain matrix $A \in \mathbb{M}$ is Hurwitz.

    \item[(A3)] Each uncertain matrix $B \in {\mathbb{B}} = { \text{co}}(B_i)_{i \in \mathbb{K}}$ with $B_i \in \Rf^{n \times m}, \forall i \in \Kn$ is full row rank. 

    \item[(A4)] The disturbance $f(t)$ is a smooth function, satisfying the following inequality
    \begin{align}
    \label{eq02}
    \| \dot{f}(t) \| \leq \delta, \quad \forall t \geq 0
    \end{align}
    
where $\delta$ is a positive constant.
\end{description}

Assumption (A1) implies that any element of $\mathbb{M}$ is given by $(A, E, C, D, B) \!=\! \sum_{i \in \mathbb{K}} \lambda_i( A, E, C, D, B)_i$ for some $\lambda \in \Lambda$. Assumption (A2) ensures that the zero dynamics given by $\dot{\zeta} = A \zeta$ is globally exponentially stable. Notice that this condition is usual since the sliding surface can be appropriated defined to ensure that this assumption is fulfilled, see \cite{U:92} and the example given in Section $7$. Assumption (A3) leads to ${\rm rank}(B) = n, \forall B \in \mathbb{B}$. The last assumption is usually adopted in the framework of MSTA control design, see \cite{gero:2023}.

\section{Robust MGSTA Lyapunov control design}

The main purpose of this paper is to provide an LMI-based procedure for MGSTA controller design, extending its applicability to a broader class of linear multivariable dynamic systems. 

Naturally, we consider that system (\ref{eq:main_system}) can be viewed as a feedback connection of the linear, time-invariant subsystem $\mathcal{S}_L$, with state space realization
\begin{align} 
\label{eq17} \dot \zeta & = A\zeta + E \sigma \\
\label{eq18} y & = C \zeta  + D \sigma 
\end{align} 
and the nonlinear subsystem $\mathcal{S}_{N}$ that emerges from the synthesis of the MSTA control, that is 
\begin{align} 
\label{eq19} \dot \sigma & = y + B u + f \\
\label{eq21}   \dot \eta &  = c(\sigma)^2 \sigma \\
\label{eq20}   u & = K_0 \zeta + c(\sigma) K_1 \sigma + K_2 \eta  
\end{align}
The matrices $K_1, K_2 \in \mathbb{R}^{m \times n}$
are the MSTA control gains, while $K_0 \in \mathbb{R}^{m \times r} $
is an additional gain introduced to facilitate the handling of internal dynamic coupling. 
{\color{black}{The positive scalar function $c(\sigma)$ is given by}}
 \begin{align} \label{eq22}
c(\sigma) = \frac{1}{\sqrt{\|\sigma\|}}  + \alpha ~:~\mathbb{R}^n \rightarrow \mathbb{R}_+, \ \ \ \alpha > 0 
\end{align}
and is well-defined and differentiable for all $\sigma \neq 0$.

Assuming, with no loss of generality, see \cite{gero:2023}, that matrix $BK_2 \in \Rf^{n \times n}$ is non singular, applying the change of variables 
\begin{equation} \label{eq23}
z = \eta +  (B K_2)^{-1} f
\end{equation}
it is readily verified that, defining the state variable $x =[\sigma^T~z^T]^T \in \Rf^{2n}$, the nonlinear subsystem $\mathcal{S}_{N}$, has the state space realization
\begin{align}
\label{eq24}  \dot x & =  L(\sigma) A_K R(\sigma) x + L(\sigma) F_0 w + G_0^T y  \\
\label{eq25} \sigma & = G_0x
\end{align}
where $w(t) = c(\sigma)^{-1} (BK_2)^{-1} \dot f(t)$. Moreover, the matrices $L(\sigma) = {\textrm{diag} }(I_n,c(\sigma)I_n )$, $R(\sigma) = {\textrm{diag}}(c(\sigma) I_n,I_n)$ are diagonal and well-defined for all $\sigma \neq 0$, $F_0 = [0~I_n]^T \in \Rf^{2n \times n}$, $G_0 = [I_n~0] \in \Rf^{n \times 2n}$ and, as expected, the equality $A_K =A_0 + B_0 K$ is established by denoting
\begin{align}
	\label{eq26}  A_0 \!=\! \left [ \begin{array}{cc} 0 & 0 \\ I_n & 0 \end{array} \right ],~B_0 \!=\! \left [ \begin{array}{c} B \\ 0  \end{array} \right ],~K \!=\! \left [ \begin{array}{cc} K_1 & K_2 \end{array} \right ]
\end{align}
It is important to stress that the factorization leading to matrices $L(\sigma)$ and $R(\sigma)$ are useful for analysis purposes. The solution of \eqref{eq24}
are understood in the sense of Filippov \cite{F:88}. Furthermore, considering all these equations as a whole, eventual difficulties of a possible lack of definition for $\sigma = 0$ do not occur. Note that the disturbance $w \in \Rf^n$  has an important property, namely: for a given $\gamma>0$, if the norm bounded constraint \eqref{eq02} is satisfied for $\delta = \|(BK_2)^{-1}\|^{-1} \gamma^{-1}$, then $w^Tw \leq  c(\sigma)^{-2} \gamma^{-2}$. Notice that to take into account the control gain $K_0$, instead of (\ref{eq18}), the output $y$ becomes 
\begin{align} \label{eqynew}
y = C_{0}\zeta + D \sigma
\end{align}
with $C_{0} = C + B K_0$. Now, let us introduce the following controlled output variable $ \xi = (H + J K_0) \zeta$, where the matrices $H \in \Rf^{q \times r}$ and $J \in \Rf^{q \times m}$ select and weight each state and control variables (or a linear combination of them) constituting the final cost. Thus, as in the classical quadratic optimal control, we can interpret $H$ as the way to build up an appropriate cost associated with the state transient, prior to reach the sliding surface and during the sliding motion. Matrix $J$ penalizes the control effort. Hence, by choosing $H$ and $J$ it is possible to restrict the transient velocity for each component of the state variable of the linear subsystem $\mathcal{S}_L$ keeping the control gain within reasonable bounds. This cost definition and its consequence will be addressed with more details afterwards. 

\subsection{Lyapunov function candidates}

The linear subsystem $\mathcal{S}_L$ as well as the nonlinear subsystem $\mathcal{S}_{N}$ are analysed from the viewpoint of two different Lyapunov functions. The difficulty of obtained, from them, a Lyapunov function candidate to their feedback connection is investigated and discussed afterwards. For the moment, let us first introduce the Lyapunov function candidates, namely 
\begin{align}
	\label{eq27}  v(\zeta) = \zeta^T S \zeta 
\end{align}
with $S \in \Rf^{r \times r}$ being a symmetric, positive definite matrix associated to $\mathcal{S}_L$ and 
\begin{align}
	\label{eq28}  V(x) = x^TR(\sigma) P R(\sigma) x
\end{align}
with $P \in \Rf^{2n \times 2n}$ being a symmetric, positive definite matrix, associated to $\mathcal{S}_N$. It is readily seen that there exist positive scalars $\varepsilon_L$ and $\theta_L$ such that $\varepsilon_L \|\zeta\|^2 \leq v(\zeta) \leq \theta_L \|\zeta\|^2$, for all $\zeta \in \Rf^r$. Likewise, there exist positive scalars $\varepsilon_N$ and $\theta_N$ such that $\varepsilon_N \|R(\sigma) x\|^2 \leq V(x) \leq \theta_N \|R(\sigma) x\|^2$, for all $x \in \Rf^n$. Together, these bounds imply that both functions are radially unbounded.

\subsection{The linear subsystem $\mathcal{S}_L$}

Since $\mathcal{S}_L$ is a time-invariant linear system, defined by the state space equations (\ref{eq17}) and (\ref{eqynew}), it is natural to consider the quadratic Lyapunov function candidate (\ref{eq27}) whose time derivative with respect to an arbitrary trajectory is given by 
\begin{align}
	\label{eq29}  \dot v(\zeta) & = 2 \zeta^T S (A \zeta + E \sigma) \nonumber \\
 & = \zeta^T (A^T S + S A ) \zeta + 2 \zeta^T S E G_0 x \,,
\end{align}
where we have used (\ref{eq25}). The next lemma is an instrumental result to be used in the sequel. 
\begin{lemma} \label{lemma01}
If the matrix inequality
\begin{align} 
\label{eq30} \left [ \begin{array}{cc}  A^TS + S A + \rho S & \bullet \\ \alpha^{-1} G_0^TE^TS  & - \rho P \end{array} \right ] < 0
\end{align}
is satisfied for scalars $\alpha>0$, $\rho>0$, and matrices $S>0$, $P>0$ of appropriate dimensions, then the inequality
\begin{align} 
\label{eq31} \dot v & < \rho  ( V - v )
\end{align}
holds, for all $0 \neq \zeta \in \Rf^r$ and $0 \neq x \in \Rf^{2n}$.
\end{lemma}
\begin{proof}
The Schur Complement with respect to the second row and column of (\ref{eq30}) together with $0 < c(\sigma)^{-2} < \alpha^{-2}$ shows that 
\begin{align} 
\label{eq32} \left [ \begin{array}{cc}  A^TS + S A + \rho S & \bullet \\ c(\sigma)^{-1} G_{0}^TE^TS  & - \rho P \end{array} \right ] < 0
\end{align}
holds as well. Multiplying it to the left by $[\zeta^T~x^TR(\sigma)]$ and to the right by its transpose, taking into account that $c(\sigma)^{-1}G_{0}R(\sigma) x = G_{0}x$, (\ref{eq32}) yields
\begin{align} 
\label{eq33} \dot v & = \zeta^T (A^T S + S A ) \zeta + 2 \zeta^T S E G_{0}x \nonumber \\
& = \zeta^T (A^T S + S A ) \zeta + 2 c(\sigma)^{-1}\zeta^T S E G_{0}R(\sigma) x \nonumber \\
& < - \rho \zeta^T S \zeta + \rho x^TR(\sigma)PR(\sigma) x
\end{align}
which reproduces (\ref{eq31}). The proof is concluded. 
\end{proof} 

It is interesting to observe that a necessary condition for the feasibility of the inequality (\ref{eq30}) is matrix $A$ to be Hurwitz stable. This means that the isolated (with $\sigma=0$) subsystem $\mathcal{S}_L$ must be asymptotically stable. Or more generally, in the presence of parameter uncertainty, it must be robustly stable. These conditions are ensured by
Assumption (A3). In a practical setting, such conditions can be achieved through a proper sliding surface design, as illustrated in Section~\ref{trailer}. Finally, it is also important to point out that (\ref{eq31}) imposes that $\dot v(\zeta)<0$ only in a subset and not for all $(\zeta, x) \in \Rf^r \times \Rf^{2n}$.

\subsection{The nonlinear subsystem $\mathcal{S}_N$}

We are now in position to deal with the nonlinear subsystem $\mathcal{S}_N$, defined by equations (\ref{eq24})-(\ref{eq25}). The algebraic manipulations are much more involved and some difficulties must be overcome. In order to perform the time derivative of $V(x)$ along an arbitrary trajectory of $\mathcal{S}_N$, let us determine
\begin{align}
	\label{eq34} R(\sigma) \dot x & = R(\sigma) \big ( L(\sigma) A_K R(\sigma) x + L(\sigma) F_0 w  +  G_0^T y \big ) \nonumber \\
	& = c(\sigma) \big ( A_K R(\sigma)x + F_0 w + G_0^T y \big )
\end{align}
and also, see \cite{gero:2023}
\begin{align} \label{eq35}
\dot R(\sigma) x  & =  - \left ( \frac{1}{\sqrt{\|\sigma\|}} \right ) E_0 \Pi_\sigma \big ( BK R(\sigma) x + y \big )
\end{align}
in which expression $E_0 = [(1/2)I_n~0]^T \in \Rf^{2n \times n}$ and $\Pi_\sigma = (\sigma /  \|\sigma\|) (\sigma /  \|\sigma\|)^T \in \Rf^{n \times n}$ is a rank-one matrix with the key property $ 0 \leq \Pi_\sigma \leq I_n$ for all $0 \neq \sigma \in \Rf^n$. Hence, with (\ref{eq34})-(\ref{eq35}) and performing the factorization of the output variable $y$, the time derivative of the Lyapunov function along with an arbitrary trajectory of $\mathcal{S}_N$, can be expressed as 
\begin{align}
\label{eq36} \dot  V(x) & = 2 c(\sigma) x^T R(\sigma)P \big ( A_K R(\sigma)x + F_0 w \big )  \nonumber \\
& ~~~~ - \left ( \frac{2}{\sqrt{\|\sigma\|}} \right ) x^TR(\sigma)P E_0 \Pi_\sigma  BK R(\sigma) x \nonumber \\
& ~~~~ + 2 c(\sigma)  x^TR(\sigma) P \Gamma_\sigma y
\end{align}
where, as it can be verified, the matrix $\Gamma_\sigma$ is given by $\Gamma_\sigma = G_0^T - \rho_\sigma E_0 \Pi_\sigma$ with $\rho_\sigma  = 1 / \big ( c(\sigma)\sqrt{\|\sigma\|} \big )$. Fortunately, from the definition (\ref{eq22}), the scalar variable is such that $\rho_\sigma \in (0, 1)$ for all $\sigma \neq 0$ which allows the determination of the upper bound  
\begin{align}
\label{eq37} 0 & \leq  \Gamma_\sigma \Gamma_\sigma^T   \nonumber \\
& = \left [ \begin{array}{c} I_n - (\rho_\sigma/2) \Pi_\sigma \\ 0 \end{array} \right ] \left [ \begin{array}{c} I_n - (\rho_\sigma/2) \Pi_\sigma \\ 0 \end{array} \right ]^T \nonumber \\
& \leq G_0^T G_0
\end{align}
valid for all $\sigma \neq 0$. In addition, following the same path, simple algebraic manipulations establish the useful upper bound for the disturbance, that is 
\begin{align}
\label{eq38} w^Tw  & \leq  c(\sigma)^{-2} \gamma^{-2}    \nonumber \\
& =  \left ( \frac{c(\sigma)^{-2} \gamma^{-2}}{c(\sigma)^2 \|\sigma\|^2} \right ) x^T R(\sigma)G_0^TG_0R(\sigma)x \nonumber \\
& = \left ( \rho_\sigma^4 \gamma^{-2} \right )  x^T R(\sigma)G_0^TG_0R(\sigma)x \nonumber \\
& \leq \gamma^{-2} x^T R(\sigma)G_0^TG_0R(\sigma)x
\end{align} 
Based on these algebraic relationships we are in a position to calculate a suitable upper bound to the time derivative of the Lyapunov function candidate that follows from (\ref{eq36}). Indeed, usual reasoning indicates that the upper bound 
\begin{align}
\label{eq39} -PE_0\Pi_\sigma BK & - K^T B^T \Pi_\sigma E_0^T P \nonumber \\
& \leq \mu PE_0E_0^TP + \mu^{-1}K^TB^TBK
\end{align}
holds for all $\mu > 0$, and together with (\ref{eq38}), the inequality
\begin{align}
\label{eq40} 2x^TR(\sigma)PF_0w & \leq x^TR(\sigma)\Big ( \beta PF_0F_0^TP \nonumber \\
& ~~~~~~ + \beta^{-1} \gamma^{-2} G_0^TG_0 \Big ) R(\sigma) x
\end{align}
holds for all $\beta >0$. Finally the inequality
\begin{align}
\label{eq41} 2x^TR(\sigma)P\Gamma_\sigma y & = 2x^TR(\sigma)P\Gamma_\sigma (C_{g0}\zeta + D_g \sigma ) \nonumber \\
& \leq \pi x^TR(\sigma) PG_0^TG_0P  R(\sigma) x \nonumber \\ & ~~~+ \pi^{-1} \sigma^TD_g^TD_g \sigma \nonumber \\
& ~~~+ \kappa  x^TR(\sigma)PG_0^TG_0PR(\sigma)x \nonumber \\ & ~~~+   \kappa^{-1} \zeta^TC_{g0}^TC_{g0}\zeta
\end{align}
also holds for all scalars $\pi>0$ and $\kappa>0$. Hence, taking into account that $\sigma = c(\sigma)^{-1} G_{0} R(\sigma)x$, these bounds together with inequality (\ref{eq36}) yield
\begin{align}
\label{eq42} \dot  V(x) & \leq  c(\sigma) x^T R(\sigma) \Big ( A_K^T P + P A_K + \mu PE_0E_0^TP \nonumber \\
& + \mu^{-1}K^TB^TBK + \beta PF_0F_0^TP \nonumber \\
& + \beta^{-1} \gamma^{-2} G_0^TG_0 + \pi PG_0^TG_0P  \nonumber \\
& + \pi^{-1} \alpha^{-2} G_{0}^TD_g^TD_gG_{0} + \kappa PG_0^TG_0P  \Big ) R(\sigma)x \nonumber \\
&  + c(\sigma) \kappa^{-1} \zeta^TC_{g0}^TC_{g0}\zeta
\end{align}
This inequality, which does not impose $\dot V(x)$ negative for all $(\zeta, x) \in \Rf^r \times \Rf^{2n}$ is on the core of the stability conditions of the whole system composed by the feedback connection of the linear subsystem $\mathcal{S}_L$ and the nonlinear subsystem $\mathcal{S}_N$. The result stated in the next lemma is a consequence of the inequality (\ref{eq42}). 
\begin{lemma} \label{lemma02}
Let the control gains $K_0$ and $K$ be given. If the matrix inequalities
\begin{equation}
	\begin{aligned}
		\label{eq43}  & A_K^T P + P A_K + P \left [ \begin{array}{c} E_0^T \\ F_0^T \\ G_0 \\ G_0 \end{array} \right ]^T Z_d \left [ \begin{array}{c} E_0^T \\ F_0^T \\ G_0 \\ G_0 \end{array} \right ] P + \\ & + \left [ \begin{array}{c} B K \\ \gamma^{-1} G_0 \\  \alpha^{-1}D_gG_{0} \\ 0 \end{array} \right ]^T Z_d^{-1} \left [ \begin{array}{c} BK \\ \gamma^{-1}G_0 \\ \alpha^{-1}D_gG_{0} \\ 0 \end{array} \right ] + \rho P < 0
	\end{aligned}
\end{equation}
and
\begin{equation}
	\begin{aligned}
		\label{eq44}  \left [ \begin{array}{cc} \rho S &  \bullet \\ C_{g0} & \kappa I_n  \end{array} \right ] > 0
	\end{aligned}
\end{equation}
are satisfied for scalars $\alpha >0$, $\gamma >0$, $\rho>0$, a four-block matrix $Z_d = {\rm diag}(\mu I_n, \beta I_n, \pi I_n, \kappa I_n)>0$, and matrices $S>0$, $P>0$ of appropriate dimensions then the inequality
\begin{align}
		\label{eq45}  \dot V < \rho c(\sigma) ( v - V )
\end{align}
holds, for all $0 \neq \zeta \in \Rf^{r}$ and $0 \neq x \in \Rf^{2n}$.
\end{lemma}
\begin{proof}
With (\ref{eq43}) and (\ref{eq44}), the inequality (\ref{eq42}) becomes  
\begin{align}
		\label{eq46}  \dot V & < -c(\sigma) \rho V(x) +  c(\sigma) \kappa^{-1} \zeta^TC_{g0}^TC_{g0}\zeta \nonumber \\
  & < -c(\sigma) \rho V(x) +  c(\sigma) \rho v(\zeta) 
\end{align}
which reproduces (\ref{eq45}), proving thus the claim. 
\end{proof}
The last two lemmas are very important to establish the global stability of the whole system under consideration. The interesting feature is that the bounds (\ref{eq31}) and (\ref{eq45}) on the time derivative of the Lyapunov functions candidates can be imposed through matrix inequalities that, as it will be shown afterwards, are converted into LMIs by a suitable one-to-one change of variables. 

The reader certainly noted the fact that the bound on the time derivative (\ref{eq45}) depends explicitly on the state variable $x$ trough the positive function $c(\sigma)$, a feature that needs to be faced with care. To make this point clearer, from (\ref{eq31}) and (\ref{eq45}) it is readily verified that 
\begin{align} 
\label{eq47} \dot v(\zeta) + c(\sigma)^{-1} \dot V(x) < 0
\end{align}
for all vectors $(0,0) \neq (\zeta, x) \in \Rf^{r} \times \Rf^{2n}$. Unfortunately, the state-dependent term $c(\sigma)$ in (\ref{eq47}) prevents the determination of a Lyapunov function associated to the closed-loop system. A natural candidate would be $\nu(\zeta,x)$ such that $\dot \nu(\zeta, x) = \dot v(\zeta) + c(\sigma)^{-1} \dot V(x) < 0$ which results by solving the equations $\nabla_{\zeta} \nu = \nabla_\zeta v$ and $\nabla_{x} \nu = c(\sigma)^{-1} \nabla_x V$. The first is simple but the second one does not appear to exhibit a known solution that yields $\nu(\zeta, x)>0$ as required to be a Lyapunov function associated to the whole feedback system under analysis. On the other hand, the positive function $\nu(\zeta,x)= v(\zeta) + c(\sigma)^{-1} V(x)>0$ does not satisfy the second gradient condition and consequently it is not possible to be sure that $\dot \nu(\zeta,x)<0$. A simple and in some sense unexpected way, based on the results of Lemma \ref{lemma01} and Lemma \ref{lemma02}, to circumvent this difficulty is stated in the next theorem.  
\begin{theorem} \label{theorem02}
If there exist positive definite radially unbounded Lyapunov functions $v(\zeta)$ and $V(x)$ such that 
\begin{align} 
\label{eq48} \dot v & < \rho_L  ( V - v )\\
\label{eq49} \dot V  & < \rho_N ( v - V )
\end{align}
where $\rho_L(\zeta)>0$, $\rho_N(x)>0$, for all $0 \neq \zeta \in \Rf^{r}$ and  all $0 \neq x \in \Rf^{2n}$, along arbitrary trajectories of the subsystems $\mathcal{S}_L$ and $\mathcal{S}_{N}$, respectively. Then, the closed-loop system is globally asymptotically stable. 
\end{theorem}
\begin{proof} To prove the claim, let us consider the following Lyapunov function candidate, constructed from the Lyapunov functions (\ref{eq27}) and (\ref{eq28}), that is
\begin{align} 
\label{eq50}  \nu(\zeta, x) = \max\{ v(\zeta), V(x) \} 
\end{align}
which is immediate to verify that it is positive definite and radially unbounded. It is continuous, but it is not differentiable everywhere in $\Rf^{r} \times \Rf^{2n}$. Danskin's theorem, see \cite{Lasdon_book:1970} for details, applies and provides the one-sided directional derivative 
\begin{align} 
\label{eq51}  D_+ \nu = \left \{ \begin{array}{ccc} \dot v &,& v > V \\  \dot V &,& v < V \\  \max\{\dot v, \dot V \} &,& v = V \end{array} \right .   
\end{align}
In the region of the joint state space defined by $v(\zeta) > V(x)$ we have $\nu(\zeta,x) = v(\zeta)$ which together with (\ref{eq31}) and (\ref{eq51}) yields $D_+ \nu(\zeta,x) = \dot v(\zeta) < 0$. Likewise, in the region defined by $v(\zeta) < V(x)$ we have $\nu(\zeta,x) = V(x)$ which together with (\ref{eq45}) and (\ref{eq51}) yield $D_+ \nu(\zeta,x) = \dot V(x) < 0$. Finally in the region defined by $v(\zeta) = V(x)$, both time-derivatives in (\ref{eq31}) and (\ref{eq45}) are negative and consequently (\ref{eq51}) yields $D_+ \nu(\zeta,x)<0$ as well. This shows that the $(\zeta, x) \rightarrow (0,0)$ as $t$ goes to infinity, that is, the trajectories of the system under consideration converge to the origin whatever initial condition. The claim is proven.    
\end{proof}

The time behavior of the interconnected system can be assessed by means of the Lyapunov function $\nu(\zeta,x)$ given in (\ref{eq50}) and its one-sided derivative (\ref{eq51}). With this Lyapunov function, the next Corollary shows that the elapsed time to transit from any initial condition to the origin is not finite. Instead, what is finite is the elapsed time necessary to reach, from any initial condition outside it, a certain residual set $\mathbb{X}$ containing the origin. This is an expected behavior, 
which resembles what occurs in sliding mode control, when a trajectory hits the sliding surface and 
then evolves according to a linear differential equation, naturally exhibiting exponential decay. In fact, by considering a different definition for the variable $z$ it is possible to show that the origin of the nonlinear system $\mathcal{S}_{N}$ with a corresponding modified state $\bar x$ is reached in finite time. This result will be discussed in Section~\ref{second_order_SM}.
\begin{corollary} \label{corollary1}
Assume that the conditions of Lemmas \ref{lemma01} and \ref{lemma02} are satisfied. There exists a nonempty residual set $\mathbb{X} \subset \Rf^{r} \times \Rf^{2n}$ with the following properties i) $(0,0)$ lies in its interior and ii) it is attained in finite time from any initial condition $(\zeta_0, x_0) \notin \mathbb{X}$.  
\end{corollary}
\begin{proof} Since, by assumption, the matrix inequalities of  Lemmas \ref{lemma01} and \ref{lemma02} are feasible, there exist scalars $\varepsilon_L>0$ and $\varepsilon_N>0$, see \cite{gero:2023} for more details, such that the Lyapunov functions $v(\zeta)$ and $V(x)$ satisfy 
\begin{align} 
\label{eq52} \dot v & \leq -\varepsilon_L v +  \rho_L  ( V - v )\\
\label{eq53} \dot V  & \leq  - \varepsilon_N \sqrt{V} + \rho_N ( v - V )
\end{align}
and, consequently, the one-sided derivative (\ref{eq51}) becomes 
\begin{align} 
\label{eq54}  D_+ \nu \leq \left \{ \begin{array}{ccc} -\varepsilon_L \nu &,& \nu=v > V \\  - \varepsilon_N \sqrt{\nu}&,& v < V = \nu \\  \max\{-\varepsilon_L \nu, - \varepsilon_N \sqrt{\nu} \} &,& \nu = v = V \end{array} \right .   
\end{align}
Defining $\nu=\nu_* = (\varepsilon_N/\varepsilon_L)^2>0$, the value for which both functions $-\varepsilon_L \nu$ and $- \varepsilon_N \sqrt{\nu}$ are equal, taking into account the intervals where one is greater than the other, (\ref{eq54}) simplifies to
\begin{align} 
\label{eq55}  D_+ \nu \leq \left \{ \begin{array}{ccc} -\varepsilon_L \nu &,& 0 \leq \nu < \nu_* \\  - \varepsilon_N \sqrt{\nu}&,& \nu_* \leq  \nu \end{array} \right .   
\end{align}
Now, defining the residual set $\mathbb{X}$ as the level set of the function $\nu$ corresponding to $\nu = \nu_*$, that is 
\begin{align} 
\label{eq56}  \mathbb{X} = \{ (\zeta, x) : \nu(\zeta,x) \leq \nu_* \} 
\end{align}
then, by definition, one may conclude that it exhibits the properties i) and ii). Moreover, taking an arbitrary initial condition $(\zeta_0, x_0) \notin \mathbb{X}$, the integration of (\ref{eq55}) from $t=0$ to $t>0$ proves that this set is attained in a finite time satisfying $T(\nu_0) \leq (2/\varepsilon_N) ( \sqrt{\nu_0} - \sqrt{\nu_*})$. The proof is concluded. 
\end{proof}

This result is useful in many instances of interest. In principle, the Lyapunov function $\nu(\zeta,x)$ given in (\ref{eq50}) can be designed such that the level set $\mathbb{X}$ defined in (\ref{eq56}) is {\em small} with respect to some metric. Another possibility is to determine the Lyapunov function such that a performance index of interest has a guaranteed minimum upper bound facing parameters and exogenous uncertainties. This last proposition, more realistic in our opinion, is addressed in the next section.

\section{Robust performance analysis}

The guaranteed performance of the closed-loop system can be defined as follows: For a given scalar $\gamma>0$ and the state feedback gains $K_0 \in \Rf^{m \times r}$, $K \in \Rf^{m \times 2n}$, determine the smaller guaranteed performance index $\varrho_{rob}$ such that 
\begin{align} 
\label{eq57}  \sup_{\|w\|^2 \leq c(\sigma)^{-2} \gamma^{-2}} \int_0^\infty \xi^T \xi dt \leq \varrho_{rob}^2
\end{align}
where, as discussed before, the controlled output of interest is expressed as $\xi = (H+ JK_0 ) \zeta$. This quadratic performance index penalizes the state trajectory deviations from zero of the linear subsystem $\mathcal{S}_L$, imposing thus good transient response. To this end, let us replace the inequality (\ref{eq30}) by
\begin{align} 
\label{eq58} \left [ \begin{array}{ccc}  A^TS + S A + \rho S & \bullet & \bullet \\ \alpha^{-1} G_0^TE^TS & - \rho P & \bullet \\ H + JK_0 & 0 & - I_q \end{array} \right ] < 0
\end{align}
and inequality (\ref{eq44}) by 
\begin{align}
		\label{eq59}  \left [ \begin{array}{ccc} \rho S &  \bullet & \bullet \\ C_{0} & \kappa I_n & \bullet \\ H + JK_0 & 0 & \alpha I_q \end{array} \right ] > 0
\end{align}
which allow us to state and prove the next theorem that provides the minimum guaranteed performance index $\varrho_{rob}$ satisfying the upper bound (\ref{eq57}). 
\begin{theorem} \label{theorem04}
Let the scalar $\gamma>0$, the state feedback gains $K_0$, $K$ and a non null initial condition $(\zeta_0, x_0)$ be given. Assuming feasibility, the minimum robust guaranteed performance index equals to 
\begin{align}
		\label{eq60}  \varrho_{rob}^2 = \inf_{\alpha>0, \rho>0, Z_d>0, S, P} \{ \nu(\zeta_0, x_0) : (\ref{eq43}), (\ref{eq58})-(\ref{eq59}) \} 
\end{align}
\end{theorem}
\begin{proof} The Schur Complement with respect to the third row and column of inequality (\ref{eq58}), the result of Lemma \ref{lemma01}, expressed by (\ref{eq31}), becomes 
\begin{align}
		\label{eq61}   \dot v < \rho(V - v) - \xi^T \xi
\end{align}
Likewise, assuming that (\ref{eq43}) holds, performing the Schur Complement with respect to the third row and column of inequality (\ref{eq59}), the result of Lemma \ref{lemma02}, expressed by (\ref{eq45}), becomes 
\begin{align}
		\label{eq62}   \dot V & < \rho c(\sigma)(v - V) - c(\sigma) \alpha^{-1} \xi^T \xi \nonumber \\
  & \leq \rho c(\sigma)(v - V) - \xi^T \xi
\end{align}
where, the second inequality follows from the fact that (\ref{eq22}) implies $c(\sigma) \alpha^{-1} \geq 1$ for all $\sigma \neq 0$ and all $\alpha >0$. Hence, these bounds on the time derivative of the Lyapunov functions $v$ and $V$, provide the one-sided directional derivative, see (\ref{eq51}), satisfying $D_+ \nu \leq - \xi^T \xi$, for all $(\zeta, x)$. As a consequence, since by assumption, problem (\ref{eq60}) is feasible, by Theorem \ref{theorem02} the closed-loop system is globally asymptotically stable, in which case the time integration from $0$ to $+\infty$ yields 
\begin{align} 
\label{e63}  \sup_{\|w\|^2 \leq c(\sigma)^{-2} \gamma^{-2}} \int_0^\infty \xi^T \xi dt \leq \nu(\zeta_0, x_0)
\end{align}
making clear that $\varrho_{rob}^2$ given in (\ref{eq57}) is the smallest upper bound. The proof is concluded. 
\end{proof}

As usual, a simple modification of the stability conditions suffices to obtain the best guaranteed performance index that can be granted by the class of Lyapunov function candidates. Due to the complexity of the control design problem under consideration, it is a bit surprising that problem (\ref{eq60}) can be shown to be expressed by LMIs with respect to all variables (including the state feedback gains) but the pair of positive scalar variables $(\alpha, \rho)$. This is the subject of the next section. 

\section{LMI-based conditions}

In this section, our goals are twofold. First, we show that problem (\ref{eq60})  can be converted into a jointly convex programming problem even if the state feedback gains are included in the set of variables provided that that pair of positive variables $(\alpha, \rho)$ is fixed. This is accomplished by means of a suitable one-to-one change of variables. Second, a numeric procedure based on a $2$-dimensional search is developed in order to determine the control design variables associated to the minimum guaranteed performance index $\varrho_{rob}$. At this stage, parametric convex bounded uncertainties defined by the convex polyhedral set $\mathbb{M}$ are also included.

Consider the one-to-one change of variables, of compatible dimensions, defined by
\begin{align} 
\label{eq64}  (S^{-1}, P^{-1}, KP^{-1}, K_0S^{-1}) ~\Longleftrightarrow~ (Q, X, Y, W)
\end{align}
Multiplying both sides of inequality (\ref{eq43}) by $P^{-1}>0$ and using the change of variables (\ref{eq64}), the LMIs
\begin{align}
\label{eq65} \left [ \begin{array}{cc} \left ( \begin{array}{c} A_0 X + B_{0i} Y \\ + X A_0^T + Y^T B_{0i}^T + \rho X  \\ + \left [ \begin{array}{c} E_0^T \\ F_0^T \\ G_0 \\ G_0 \end{array} \right ]^T Z_d \left [ \begin{array}{c} E_0^T \\ F_0^T \\ G_0 \\ G_0 \end{array} \right ] \end{array} \right )  & \bullet   \\ \left (\begin{array}{c} B_iY \\ \gamma^{-1}G_0X \\ \alpha^{-1} D_{i} G_0 X \\ 0 \end{array} \right ) & -Z_d  \end{array} \!\! \right ] <0
\end{align}
for $i \in \Kn$, imply that (\ref{eq43}) holds for all $B \in \mathbb{B}$. Indeed, multiplying (\ref{eq65}) successively by the components of $\lambda \in \Lambda$ and summing up the partial results, the inequality (\ref{eq43}) follows. Multiplying the inequality (\ref{eq58}) both sides by ${\rm diag}(S^{-1}, P^{-1}, I_q)$, the same reasoning provides the LMIs 
\begin{align}
\label{eq66} \left [ \begin{array}{ccc} \left ( \begin{array}{c} A_{i} Q  \\ + Q A_{i}^T  + \rho Q  \end{array} \right )  & \bullet & \bullet \\ \alpha^{-1} X G_0^TE_{i}^T & - \rho X & \bullet \\ HQ + JW & 0 & - I_q \end{array} \!\! \right ] <0
\end{align}
for $i \in \Kn$. Similarly, multiplying the inequality (\ref{eq59}) both sides by ${\rm diag}(S^{-1}, I_n, I_q)$, once again, similar arguments lead to the LMIs 
\begin{align}
\label{eq67} \left [ \begin{array}{ccc} \rho Q   & \bullet & \bullet \\ C_{i}Q + B_{i}W & \kappa I_n & \bullet \\ HQ + JW  & 0 &  \alpha I_q \end{array} \!\! \right ] > 0 
\end{align}
for $i \in \Kn$. As it can be easily verified, the inequalities (\ref{eq65})-(\ref{eq67}) are LMIs with respect to the matrix variables $Q$, $X$, $Y$, $W$ and $Z_d$ provided that the pair of positive variables $(\alpha, \rho)$ is fixed. Finally, keeping in mind that the maximum of two elements equals the minimum upper bound, the LMIs
\begin{align}
		\label{eq68}  \left [ \begin{array}{cc} \theta & \bullet \\ \zeta_0 & Q \end{array} \right ]>0,~\left [ \begin{array}{cc} \theta & \bullet \\ R(\sigma_0) x_0 & X \end{array} \right ]>0
\end{align}
depending on the scalar variable $\theta>0$, enable us to determine the Lyapunov function evaluated at the initial condition $(\zeta_0, x_0)$ where $x_0 = [\sigma_0^T~ \eta_0^T]^T$ whenever the exogenous disturbance satisfies $f(0)=0$, see (\ref{eq02}). The case of $ f(0) \neq 0$ can be addressed following the same steps presented in \cite{NGH:2024}. Moreover, note that $R(\sigma_0)$ relies on the positive scalar $\alpha >0$. These finds are summarized in the next theorem. 
\begin{theorem} \label{theorem05}
Let the scalars $\gamma>0$, $\alpha>0$, $\rho>0$ and a non null initial condition $(\zeta_0, x_0)$ be given. The optimal solution of the convex programming problem 
\begin{align}
		\label{eq69}  \inf_{\theta, Q, X, Y, W, Z_d} \{ \theta : (\ref{eq65})-(\ref{eq68}) \} 
\end{align}
provides the state feedback matrix gains $K_0 = WQ^{-1}$ and $K = Y X^{-1}$ such that the closed-loop system is globally asymptotically stable with guaranteed performance index $\theta$, whenever the exogenous disturbance satisfying (\ref{eq02}) and the parametric uncertainty defined by the polyhedral set $\mathbb{M}$ is concerned. 
\end{theorem}
\begin{proof} The proof is a consequence of  Theorem~\ref{theorem02}. Whenever (\ref{eq69}) is feasible the corresponding state feedback gain ensure global robust asymptotic stability of the closed-loop system. Moreover, since $\alpha>0$ and $\rho>0$ are fixed (but not necessarily optimal) the value of $\theta$ is an upper bound to the guaranteed performance index, because 
\begin{align}
	\label{eq70}  \nu(\zeta_0, x_0) & = \max  \{ v(\zeta_0), V(x_0) \} \nonumber \\
  & = \inf_{\theta} \{ \theta : (\ref{eq68}) \} 
\end{align}
concluding thus the proof. 
\end{proof}

From Theorem \ref{theorem05}, there is no difficulty to calculate the optimal state feedback gains imposing the optimal guaranteed performance to the closed-loop system, whenever subject to exogenous unknown perturbations satisfying (\ref{eq02}) and parameter convex uncertainty described by the set $\mathbb{M}$. To this end, it suffices to solve
\begin{align}
\label{eq71}  \varrho_{rob}^2 = \inf_{\alpha>0, \rho>0} \Big \{ \inf_{\theta, Q, X, Y, W, Z_d} \{ \theta : (\ref{eq65})-(\ref{eq68}) \} \Big \}
\end{align}
where it must be observed that the inner problem is convex and so can be solved for each fixed $(\alpha, \rho)$ with no difficulty. The outer minimization is performed by a suitable $2$-dimensional search procedure. In this way, the joint global optimum is reached. 

Last but not least is the fact that the gain of the MSTA is not directly constrained and so its magnitude may be expressive being, in may cases, inappropriate. A way to circumvent this occurrence is to include in the problem (\ref{eq71}) the convex constraint 
\begin{align}
		\label{eq72}  \left [ \begin{array}{cc} \omega I_m & \bullet \\ Y^T & X \end{array} \right ]>0
\end{align}
where the scalar $\omega>0$ is given. Actually, recalling that $u_{ST}(t) = K R(\sigma(t))x(t)$ is the MSTA control signal, (\ref{eq72}) yields $K^TK < \omega P$ and consequently 
\begin{align}
		\label{eq73}  \|u_{ST}(t)\|^2 & = x(t)^T R(\sigma(t))K^TKR(\sigma(t))x(t) \nonumber \\
        & < \omega V(x(t))  \nonumber \\
        &  \leq \omega \nu(\zeta(t),x(t)) \nonumber \\
        & \leq \omega \theta,~\forall t \in \Rf_+
\end{align}
indicates that the MSTA control signal magnitude, for all $t \in \Rf_+$, can the imposed by the appropriate choice of the scalar $\omega>0$.  

\section{
\label{second_order_SM}
The realization of second-order sliding modes
}

To conclude the theoretical analysis, in this section, we show that a second-order sliding mode is achieved after some finite time. To accomplish this, let us redefine the variable $z(t)$ in a more convenient manner for this purpose as
\begin{align} \label{eq_newz}
\bar z = \eta +  (B K_2)^{-1} \bigg( f +  C_{0} \zeta \bigg )
\end{align}
where, from (\ref{eq23}), $\bar z = z +  (B K_2)^{-1} C_{0} \zeta $ and since we have already proven that $\zeta$ and $z$ converge to zero, it is straightforward to see that $\bar z$ also converges to zero. 

Once again, we first consider that the system matrices are given. Subsequently, the extension to include polytopic uncertainty can be easily achieved due to the linear nature of the problem concerning the uncertainty. 
Taking \eqref{eq_newz} into account, and redefining the state variable  $\bar x =[\sigma^T~\bar z^T]^T \in \Rf^{2n}$, the nonlinear subsystem $\mathcal{S}_{N}$, can be described by
\begin{align}
    \label{eq_barx1}  \dot{\bar{x}} & =  L(\sigma) A_K R(\sigma) \bar x + L(\sigma) F_0 \bar w + G_0^T D G_0 \bar x  \\
\label{eq_barx2}  \sigma & = G_0 \bar x
\end{align}
where the modified disturbance is now given by
\begin{align}
    \label{eq_barx3}
\bar w(t) & = c(\sigma)^{-1} (BK_2)^{-1}  \dot{\bar{f}}(t) 
\end{align}
with $ \bar f(t) = f(t) + C_{0} \zeta$. Hence, considering that 
\begin{align} \label{eq_barx4}
\| \dot{\bar{f}}(t) \| & \leq  \| \dot{f}(t) \| + \|C_{0} \dot \zeta\| 
\end{align}
two important aspects regarding (\ref{eq_barx1})-(\ref{eq_barx2}) deserve some comments. First, the global asymptotic stability result previously reported, implies that $\dot{\zeta}(t)$ is a bounded signal with an upper bound that decreases over time. Second, as the dynamic equation of $\bar x$ does not depend directly on $\zeta$ (which only affects the bounded disturbance $\bar w$) the analysis can proceed without considering the linear subsystem. Actually, given $\gamma>0$ and $\delta = \|(BK_2)^{-1}\|^{-1} \gamma^{-1}$ there exists a finite time $t_s>0$ such that
\begin{align}
    \label{eq_dot_barf_bound}
    \| \dot{\bar{f}}(t) \| \leq \delta,~\forall t \geq t_s
\end{align}
and consequently $\bar w^T \bar w \leq  c(\sigma)^{-2} \gamma^{-2}, \forall t \geq t_s$, whenever the exogenous disturbance satisfies $\|\dot f(t) \| < \delta$. Adopting the same Lyapunov function as before, but now associated to the system (\ref{eq_barx1})-(\ref{eq_barx2}),   
\begin{align}
\label{eq_barV}  V( \bar x) = \bar x^TR(\sigma) P R(\sigma) \bar x
\end{align}
it can be verified that inequality \eqref{eq43} is satisfied for all $t \geq t_s$, which in turn implies that
\begin{align} \label{eq_dbarV}
\dot V( \bar x) \leq - c(\sigma) \| R(\sigma) \bar x \|^2,~\forall t \geq t_s
\end{align}
Taking into account that $c(\sigma) \geq 1/\sqrt{ \|\sigma \|}$ for all $\sigma \neq 0$ and $\| R(\sigma) \bar x\| \geq \sqrt{\| \sigma \|}$, there exist $\theta_N>0$ and  $\varepsilon_N>0$ such that 
\begin{align} \label{eq_dbarV2}
\dot V( \bar x) \leq  - \frac{\varepsilon_N}{\sqrt{\theta_N}} V(\bar x)^{1/2}, \forall t \geq t_s
\end{align}
allowing, from usual arguments, see \cite{gero:2023}, that the equilibrium point $\bar x = 0$ is globally reached after some finite elapsed time. 
The next section is entirely devoted to the application of the theoretical results presented so far in a practical motivated control problem.

\section{Control of a chain of trailers}
\label{trailer}

Consider the fault-tolerant control of a chain of three trailers, inspired by \cite{HCC:2003}, \cite{CCH:2003} and illustrated in \figurename~\ref{fig:chainOfTrailers}. The system is controlled by actuators on the two active trailers. The output vector $q_a = [ q_{a1}~ q_{a2} ]^T$ denotes the positions of the active trailers, $q_p$ denotes the position of the passive trailer, and the control variable $u_a \in \mathbb{R}^3$ represents the actuator forces. A redundant actuator allows the system to tolerate the complete loss of one of them. The resultant actuator forces $F_a \in \mathbb{R}^2$ is described by
\begin{align*}
F_a = S_a \mathcal{F} u_a,~S_a = \left[ \begin{array}{rrr}
    1 & -1 & 0 \\ 0 &  1 & -1 \end{array} \right]
\end{align*}
where $\mathcal{F} = \text{diag}(\mathcal{F}_1,~\mathcal{F}_2,~\mathcal{F}_3)$ is a diagonal matrix indicating the fault status of each actuator. The fault index $\mathcal{F}_i$ is characterized as follows: $\mathcal{F}_i = 1$ the 
$i$-th actuator operates normally, $ 0 < \mathcal{F}_i < 1 $ the $i$-th actuator exhibits degraded performance, and $\mathcal{F}_i = 0$ the $i$-th actuator is completely non-functional.

Here, as in \cite{HCC:2003} and \cite{CCH:2003}, a passive fault-tolerant control approach is adopted. The MGSTA combined with an appropriate constant mixer matrix $S_a^T$ ensures that the closed-loop system maintains stability and desired performance, even in the presence of certain actuator faults. After applying the mixer $u_a = S_a^T u$, and considering $\zeta_t = [q_a^T~q_{p}~\dot{q}_{p}]^T \in \mathbb{R}^4$ and $\dot{q}_a  \in \mathbb{R}^2$, the motion dynamics of the trailers can be described by
\begin{align*}
    \dot \zeta_t & = \mathcal{A}_{11} \zeta_t + \mathcal{A}_{12}  \dot{q}_a \\
    \ddot{q}_a & = \mathcal{A}_{21} \zeta_t + \mathcal{A}_{22}  \dot{q}_a + B u + f_d(t) 
\end{align*}
where
\begin{align*}
\mathcal{A}_{11} & = \left[ \begin{array}{cccc} 0 & 0 & 0 & 0 \\ 0 & 0 & 0 & 0 \\ 0 & 0 & 0 & 1 \\ \frac{k_{13}}{m_3} & \frac{k_{32}}{m_3} & -\frac{k_{13}+k_{32}}{m_3} & -\frac{b_{13}+b_{32}}{m_3}  \end{array} \right], ~
\mathcal{A}_{12} = \left[ \!\! \begin{array}{cc} 1 & 0 \\ 0 & 1 \\ 0 & 0 \\ \frac{b_{13}}{m_3} &  
\frac{b_{32}}{m_3} \end{array} \!\! \right] 
\end{align*}
\begin{align*}
\mathcal{A}_{21} & \!=\! \left[ \!\! \begin{array}{cccc}
         -\frac{k_{13}}{m_1} & 0 & \frac{k_{13}}{m_1} & \frac{b_{13}}{m_1}  \\ 0 & -\frac{k_{32}}{m_2} & \frac{k_{32}}{m_2} &  \frac{b_{32}}{m_2} \end{array} \!\! \right],~
\mathcal{A}_{22} \!=\! \left[ \begin{array}{cc}
         -\frac{b_{13}}{m_1} & 0 \\
         0 & -\frac{b_{32}}{m_2} \end{array} \right]
\end{align*}
and
\begin{align*}
B &= \left[ \begin{array}{cc}
         m_1^{-1} ( \mathcal{F}_1 + \mathcal{F}_2 )  & -m_1^{-1} \mathcal{F}_2  \\
         -m_2^{-1} \mathcal{F}_2 & m_2^{-1}  ( \mathcal{F}_2 + \mathcal{F}_3 ) \end{array} \right] 
\end{align*}
The parameters are given by $m_1 = 1.0 \unit{kg}$, $2.0 \unit{kg} \leq m_2 \leq 3.0 \unit{kg}$, $2.0 \unit{kg} \leq m_3 \leq 3.0 \unit{kg}$, $k_{13} = \unit[30.0]{N/m}$, $k_{32} = \unit[45.0]{N/m}$,
$b_{13} = \unit[15.0]{Ns/m}$ and $b_{32} = \unit[30.0]{Ns/m}$, 
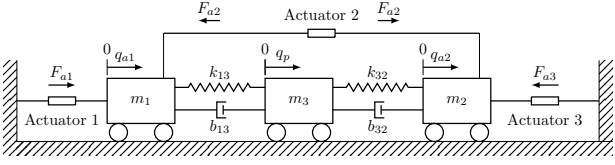
\begin{figure}[!t]
\centering
\scalebox{.6}{
\begin{tikzpicture}[every node/.style={draw,outer sep=0pt,thick}]
\tikzstyle{spring}=[thick,decorate,decoration={zigzag,pre length=0.3cm,post length=0.3cm,segment length=6}]
\tikzstyle{damper}=[thick,decoration={markings,
  mark connection node=dmp,
  mark=at position 0.5 with
  {
    \node (dmp) [thick,inner sep=0pt,transform shape,rotate=-90,minimum width=10pt,minimum height=4pt,draw=none] {};
    \draw [thick] ($(dmp.north east)+(2pt,0)$) -- (dmp.south east) -- (dmp.south west) -- ($(dmp.north west)+(2pt,0)$);
    \draw [thick] ($(dmp.north east)+(0,2pt)$) -- ($(dmp.north west)+(0,-2pt)$);
  }
}, decorate]
\tikzstyle{actuator}=[thick,decoration={markings,
  mark connection node=act,
  mark=at position 0.5 with
  {
    \node (act) [thick,inner sep=0pt,transform shape,rotate=-90,minimum width=5pt,minimum height=0.6cm,draw=none] {};
    \draw [thick] (act.north east) -- (act.south east) -- (act.south west) -- (act.north west) -- (act.north east);
  }
}, decorate]
\tikzstyle{ground}=[fill,pattern=north east lines,draw=none,minimum width=0.75cm,minimum height=0.3cm]

\node (ground) [ground,anchor=north west,minimum width=13.2cm] {};
\draw (ground.north east) -- (ground.north west);
\node (wall) [ground,anchor=north east,rotate=-90,minimum width=1.8cm] at (ground.north west) {};
\draw (wall.north east) -- (wall.north west);
\node (corner) [ground,anchor=east,minimum width=0.3cm] at (ground.west) {};
\node (wall2) [ground,anchor=north east,rotate=-90,minimum width=1.8cm] at (ground.north east) {};
\draw (wall2.south west) -- (wall2.south east);
\node (corner2) [ground,anchor=west,minimum width=0.3cm] at (ground.west) {};
\node (m1) [anchor=west,minimum width=1.5cm,minimum height=1cm,xshift=2cm] at (wall.north) {$m_1$};
\draw [thick] (m1.south west) ++ (0.25cm,-0.2cm) circle (0.2cm)  (m1.south east) ++ (-0.25cm,-0.2cm) circle (0.2cm);
\node (m3) [anchor=west,minimum width=1.5cm,minimum height=1cm,xshift=2cm] at (m1.east) {$m_3$};
\draw [thick] (m3.south west) ++ (0.25cm,-0.2cm) circle (0.2cm)  (m3.south east) ++ (-0.25cm,-0.2cm) circle (0.2cm);
\node (m2) [anchor=west,minimum width=1.5cm,minimum height=1cm,xshift=2cm] at (m3.east) {$m_2$};
\draw [thick] (m2.south west) ++ (0.25cm,-0.2cm) circle (0.2cm)  (m2.south east) ++ (-0.25cm,-0.2cm) circle (0.2cm);
\draw [actuator] (wall.north) -- node[below,draw=none,yshift=-5pt] {Actuator 1} (m1.west);
\draw [spring] (m1.north east) ++ (0cm,-0.2cm) -- node[above,draw=none,yshift=2pt] {$k_{13}$} ($(m3.north west)+(0cm,-0.2cm)$);
\draw [damper] (m1.south east) ++ (0cm,0.3cm) -- node[below,draw=none,yshift=-3pt] {$b_{13}$} ($(m3.south west)+(0cm,0.3cm)$);
\draw ($(m1.north east)+(-0.25cm,0cm)$) -- ($(m1.north east)+(-0.25cm,1cm)$);
\draw ($(m2.north east)+(-0.25cm,0cm)$) -- ($(m2.north east)+(-0.25cm,1cm)$);
\draw [actuator] ($(m1.north east)+(-0.25cm,1cm)$) -- node[above,draw=none,yshift=5pt] {Actuator 2} ($(m2.north east)+(-0.25cm,1cm)$);
\draw [spring] (m3.north east) ++ (0cm,-0.2cm) -- node[above,draw=none,yshift=2pt] {$k_{32}$} ($(m2.north west)+(0cm,-0.2cm)$);
\draw [damper] (m3.south east) ++ (0cm,0.3cm) -- node[below,draw=none,yshift=-3pt] {$b_{32}$} ($(m2.south west)+(0cm,0.3cm)$);
\draw [actuator] (wall2.south) -- node[below,draw=none,yshift=-5pt] {Actuator 3} (m2.east);
\draw [-latex,thick] ($(wall.north)+(0.7cm,10pt)$) -- node[above,draw=none] {$F_{a1}$} +(0.6cm,0);
\draw [-latex,thick] ($(m3.north)+(1.75cm,1cm+8pt)$) -- node[above,draw=none] {$F_{a2}$} +(0.48cm,0);
\draw [-latex,thick] ($(m3.north)+(-1.75cm,1cm+8pt)$) -- node[above,draw=none] {$F_{a2}$} +(-0.48cm,0);
\draw [-latex,thick] ($(wall2.north)+(-1.2cm,10pt)$) -- node[above,draw=none] {$F_{a3}$} +(-0.6cm,0);
\draw ($(m1.north west)+(0,2pt)$) -- node[at end,above,draw=none] {$0$} ($(m1.north west)+(0,12pt)$);
\draw ($(m3.north west)+(0,2pt)$) -- ($(m3.north west)+(0,12pt)$);
\draw ($(m2.north west)+(0,2pt)$) -- node[at end,above,draw=none] {$0$} ($(m2.north west)+(0,12pt)$);
\draw ($(m3.north west)+(0,2pt)$) -- node[at end,above,draw=none] {$0$} ($(m3.north west)+(0,12pt)$);
\draw [-latex,thick] ($(m1.north west)+(0,7pt)$) -- node[above,draw=none] {$q_{a1}$} +(0.8cm,0);
\draw [-latex,thick] ($(m2.north west)+(0,7pt)$) -- node[above,draw=none] {$q_{a2}$} +(0.8cm,0);
\draw [-latex,thick] ($(m3.north west)+(0,7pt)$) -- node[above,draw=none] {$q_{p}$} +(0.8cm,0);
\end{tikzpicture}
}
\caption{Chain of two active trailers and one passive trailer controlled by three actuators}
\label{fig:chainOfTrailers}
\end{figure}

The design procedure proposed in this paper was devised considering a stabilization problem. However, it can be extended to handle trajectory-tracking problems as well. In this example, our objective is to design a control law such that the output $q_a$ tracks the output $q_d$ of a reference model with transfer function $M_d(s)$ and input signal $r_d$, i.e., $q_d(s) = M_d(s) r_d(s)$. To this end, let us define the tracking errors $e_{y}  = q_{a} - q_{d}$ and $e_{v}  = \dot{q}_a - \dot{q}_d$.
To solve the trajectory tracking problem for the active trailers and reformulate the overall system in terms of tracking errors, it is essential to analyze the expected behavior of the passive trailer when trailers 1 and 2 achieve perfect tracking. Specifically, we must determine the passive trailer's desired position and velocity signals. To accomplish this, we can write $\ddot{q}_p$ as follows
\begin{align}
    \label{eq:passive_trailer}
    \ddot{q}_{p} & = -\frac{k_{13}+k_{32}}{m_3} q_{p} + \left[ \begin{array}{cc}
        \frac{k_{13}}{m_3}  & \frac{k_{32}}{m_3} \end{array} \right] (e_y + q_d) + \nonumber \\ & ~~~~-\frac{b_{13}+b_{32}}{m_3} \dot{q}_p + \left[ \begin{array}{cc}
        \frac{b_{13}}{m_3}  & \frac{b_{32}}{m_3} \end{array} \right] (e_v + \dot{q}_d)
\end{align}
Thus, when the tracking errors $e_y$ and $e_v$ converge to zero, the passive trailer's desired position and velocity signals must obey the following differential equation that depends on $q_d$ and its time derivative 
\begin{align*}
    \ddot{q}_{pd} & = -\frac{k_{13}+k_{32}}{m_3} q_{pd} + \left[ \begin{array}{cc}
        \frac{k_{13}}{m_3}  & \frac{k_{32}}{m_3} \end{array} \right] q_d \nonumber \\ & ~~~~-\frac{b_{13}+b_{32}}{m_3} \dot{q}_{pd} + \left[ \begin{array}{cc}
        \frac{b_{13}}{m_3}  & \frac{b_{32}}{m_3} \end{array} \right] \dot{q}_d
\end{align*}
It is important to stress that the desired signals for the passive trailer cannot be used for practical implementations, as they rely on knowledge of system parameters that are considered uncertain. Fortunately, due to the robustness of the MSTA only upper bounds for these signals and their derivatives are needed for the control design. From a theoretical standpoint, its use is justified in rewriting the complete system in terms of tracking errors. To achieve this, we can define the last tracking errors $e_{p1} = q_p - q_{pd}$, and $e_{p2} = \dot{q}_{p} - \dot{q}_{pd}$.

Defining the sliding surface as
\begin{equation*}
    \sigma = -\Gamma \zeta_e + e_{v}, \quad \Gamma = \left[ \begin{array}{cccc}
        -\gamma_{e1} &  0 & 0 & 0 \\
         0 & -\gamma_{e2} & 0 & 0
    \end{array} \right]
\end{equation*}
where $\zeta = [e_y^T~e_p^T]^T \in \mathbb{R}^4 $ and $\sigma \in \mathbb{R}^2$, the overall tracking error system with state variables $\zeta$ and $\sigma$, can be written as \eqref{eq:main_system} where $B$ was already given,  
\begin{align*}
A & = \left[ \begin{array}{cccc} -\gamma_{e1} & 0 & 0 & 0 \\ 0 & -\gamma_{e2} & 0 & 0 \\ 0 & 0 & 0 & 1 \\ \frac{ k_{13}-\gamma_{e1}b_{13}}{m_3} & \frac{k_{32}-\gamma_{e2}b_{32}}{m_3} & -\frac{k_{13}+k_{32}}{m_3} & -\frac{b_{13}+b_{32}}{m_3}  \end{array} \right] \\
E & = \left[ \!\! \begin{array}{cc} 1 & 0 \\ 0 & 1 \\ 0 & 0 \\ \frac{b_{13}}{m_3} &  
\frac{b_{32}}{m_3} \end{array} \!\! \right], ~
D \!=\! \left[ \begin{array}{cc}
         \gamma_{e1}-\frac{b_{13}}{m_1} & 0 \\
         0 & \gamma_{e2}-\frac{b_{32}}{m_2} 
    \end{array} \right] \\
C & \!=\! \left[ \!\! \begin{array}{cccc}
         -\gamma_{e1}^2 \!-\! \frac{k_{13}}{m_1} \!+\! \frac{ \gamma_{e1} b_{13}}{m_1} & 0 & \frac{k_{13}}{m_1} & \frac{b_{13}}{m_1}  \\ 0 & -\gamma_{e2}^2 \!-\! \frac{k_{32}}{m_2} \!+\! \frac{ \gamma_{e2} b_{32}}{m_2} & \frac{k_{32}}{m_2} &  \frac{b_{32}}{m_2} \end{array} \!\! \right]
\end{align*}
and the disturbance is defined by 
\begin{equation*}
     f(t)  = \mathcal{A}_{21} \left[ q_d^T~q_{pd}~\dot{q}_{pd} \right]^T - \ddot{q}_d(t)  + f_d(t) 
\end{equation*}
To illustrate the robustness of the MSTA we consider that actuator $1$ operates within a range from fully functional $\mathcal{F}_1 \!=\! 1$ to completely lost $\mathcal{F}_1 = 0$, i.e. ( $0 \! \leq \! \mathcal{F}_1  \leq 1)$. Taking into account, the mass uncertainties the system has $N = 8$ vertices. Note that all systems matrices are considered uncertain, and each input extreme matrix $B_i$ is non-symmetric. As expected, each uncertain matrix $A \in \mathbb{M}$ satisfies Assumption (A2).

We consider the initial states $q_a(0) = [0.04~-0.06]^T$, $\dot{q}_a(0)  =  [-0.03~0.04]^T$, and the remaining systems' initial conditions equal to zero. We have chosen the design parameters $\omega = 50$, $\gamma \!=\! 4.0$, $\gamma_{e1} = 2.0$, $\gamma_{e2} = 4.0$, and the cost matrices $H = [I_4~0]^T \in \Rf^{6 \times 4}$ and $J = [0~ I_2]^T \in \Rf^{6 \times 2}$. The reference model were selected as 
\begin{align*}
M_d(s) = \frac{125}{(s+5)^3} I_2,~ r_d(t) = \left [ \begin{array}{c} 0.10\sin(0.5t) \\ 0.06\sin(t) \end{array} \right ]
\end{align*}
and, for simulation purpose, we consider that the closed-loop system is subject to the external disturbance $f_d(t) \!=\! [\cos(t)-1~\cos(t)-1]^T$. After a $2$-dimensional search, we have roughly determined $\rho \approx 2.1$ and $\alpha \approx 11.0$, and obtained the following control gains
\begin{align*}
K_0 & = \left[ \begin{array}{rrrr}
   26.0754 & -33.4036 &  -29.2428 & -23.2136 \\
   35.1140 & -50.1011 & -13.8645 & -19.8209 
\end{array} \right] \\
K_1 & = \left[ \begin{array}{rr}
  -49.5226 & -52.2890 \\
   -5.8613 & -49.4501
\end{array} \right] \\
K_2 & = \left[ \begin{array}{rr}
 -109.3804 & -72.3834 \\
   14.8000 & -115.8515
\end{array} \right]
\end{align*}
associated to $\varrho_{rob}^2 = 583.2724$. It can be verified that the overall disturbance $f(t)$ admits the bound $\| \dot{f}(t) \| \leq 6.0$ for all $t \geq 0$. Computing $\delta = 7.8489$ from the designed gains, ensures that the disturbance $f(t)$ is feasible. 

\begin{figure}[t]
\hspace{0.7cm}\includegraphics[height=7.3cm, right]{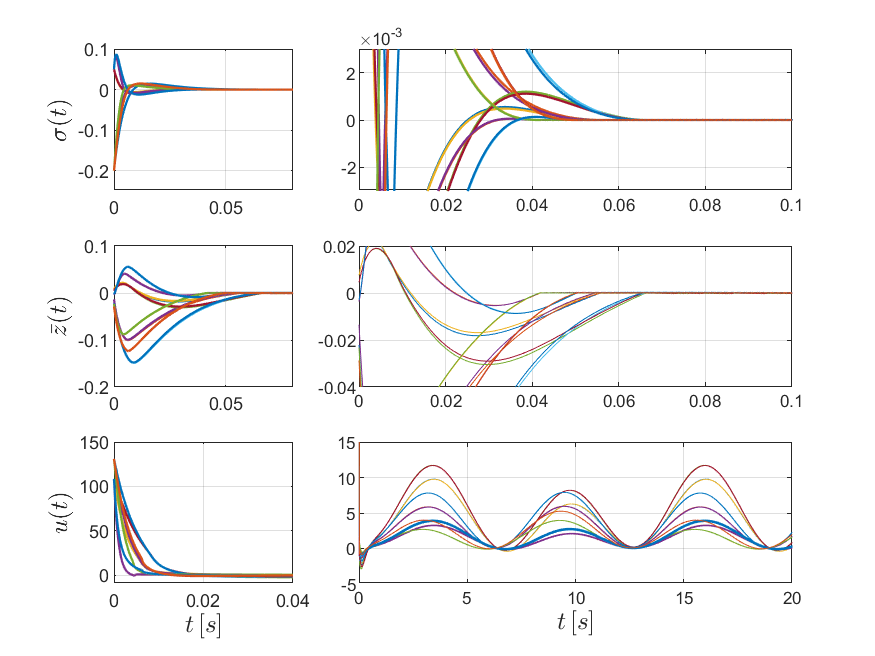}
\caption{Plot of the sliding variables $\sigma(t)$ and $\bar{z}(t)$, and of the control input $u(t)$ for all vertices.}
\label{fig03}
\end{figure}

\begin{figure}[t]
\hspace{0.7cm}\includegraphics[height=7.3cm, right]{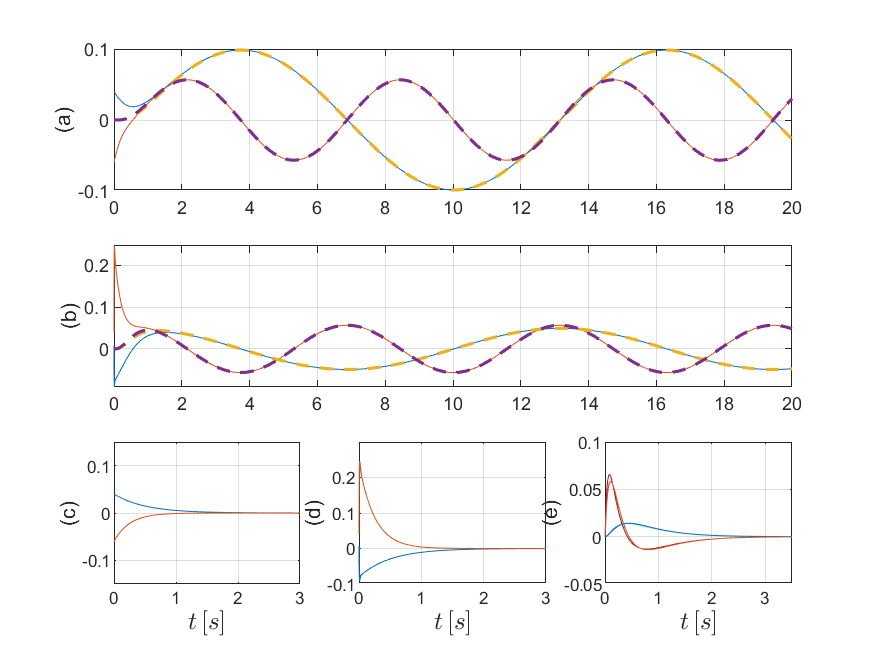}
\caption{Time evolution for all vertices of: (a) $q_a(t)$ (solid) and $q_d(t)$ (dashed); (b) $\dot{q}_a(t)$ (solid) and $\dot{q}_d(t)$ (dashed); (c) $e_y(t)$; (d) $e_v(t)$; (e) $e_p(t)$  }
\label{fig04}
\end{figure}

As illustrated in Figure~\ref{fig03}, a second-order sliding mode is achieved for each extreme matrix within the convex bounded uncertainty domain. Therefore, the result extends to any element of the compact convex polyhedral set $\mathbb{M}$, formed by the convex combination of its vertices. The maximum control effort required at the beginning is less than \unit[150]{N}, which does not appear excessively high. However, after a short time, the effort needed to achieve the desired performance is less than \unit[15]{N} in magnitude. This suggests that the MSTA appears to be very effective in rejecting exogenous disturbances despite the uncertainties and actuator faults.

As can be observed from Figure~\ref{fig04}, exact tracking is achieved as expected. It is noteworthy that the tracking errors for the position and velocity of the three trailers exhibit very similar behavior across all vertices that define the parameter uncertainty.

\section{Conclusion}

In this paper, we have introduced a novel procedure for the robust control design of linear time-invariant systems using a Multivariable Generalized Super-Twisting Algorithm (MGSTA). Our approach successfully addresses robust stability and performance conditions under the presence of convex bounded parameter uncertainty in all matrices of the plant's state-space realization, as well as exogenous  Lipschitz disturbances. These goals are accomplished through the proposal of a novel max-type non-differentiable piecewise-continuous Lyapunov function that has been shown to be well adapted to handle the class of interconnected systems under consideration.

A key feature of our closed-loop system is its sliding mode finite-time convergence, which has been thoroughly examined and validated. The design conditions were formulated as Linear Matrix Inequalities (LMIs), and we demonstrated that these can be efficiently solved using existing computational tools. Interestingly, global stability was achieved with constant controller gains despite the presence of uncertain internal dynamics of the plant.

To showcase the practical applicability and robustness of the proposed LMI-based approach, we developed a fault-tolerant MGSTA controller for a mechanical system with three degrees of freedom. The results illustrate the efficacy of our methodology, highlighting its potential for broader application in robust control design.

Overall, this work contributes to the field by providing a robust and efficient MGSTA Lyapunov/LMI design framework for general linear systems. Finite-time convergence and fault tolerance are ensured in the presence of convex  parameter uncertainties and external disturbances. Future research topics that appear as natural follow-up are full order output feedback synthesis and reaching time optimal control design.
 
\balance 
 
\begin{ack}                               
The authors would like to thank the National Council for Scientific and Technological Development (CNPq/Brazil), the Coordination for the Improvement of Higher Education Personnel (CAPES/Brazil), Finance code 001, and FAPERJ, under Grant E-26/204.669/2024, for partial financial support along the development of this research.
\end{ack}

\bibliographystyle{plain}         
\bibliography{autosam,BIB, MSTA_LMI, vsmrac}           

\end{document}